\documentclass[a4paper,UKenglish,cleveref, autoref, unicode]{lipics-v2019}

\usepackage{algorithm}
\usepackage{algorithmicx}
\usepackage{algpseudocode}
\usepackage{algpascal}
\usepackage{booktabs}
\usepackage{graphicx}

\usepackage{amsmath}
\usepackage{amssymb}
\usepackage{color}
\definecolor{red}{RGB}{255,0,0}
\definecolor{blue}{RGB}{0,0,255}
\definecolor{green}{RGB}{0,255,0}

\newcommand {\abs}[1]  {\left\vert#1\right\vert}
\newcommand {\set}[1]  {\left\{#1\right\}}

\newcommand {\defined} {\stackrel{def} {=}}

\newcommand {\bigoh}   {{\cal O}}

\newcommand {\runningtitle}[1] {\vspace{0.5ex}\noindent{\textbf{\boldmath #1:}}}

\newcommand{\ignore}[1] {}
\newcommand {\commentfig}[1] {#1}

\usepackage{tikz}

\newcommand{\probl}[3]{
\begin{flushleft}
\fbox{
\begin{minipage}{\linewidth}
\noindent {\sc #1}\\
          {\bf Input:} #2\\
          {\bf Output:} #3
\end{minipage}}
\medskip
\end{flushleft}
}

\newcommand{\cl}{{\cal L}}
\newcommand{\ce}{{\cal E}}
\newcommand{\hbm}{H-matching}
\newcommand{\mhbm}{\textsc{Max H-matching}}
\DeclareMathOperator{\parent}{par}
\DeclareMathOperator{\children}{ch}
\DeclareMathOperator{\exposed}{exp}
\DeclareMathOperator{\repr}{repr}
\DeclareMathOperator{\aug}{aug}

\title{Hierarchical b-Matching}

\author{Yuval Emek}{Department of Industrial Engineering and Management, Technion, Haifa, Israel}{yemek@technion.ac.il}{https://orcid.org/0000-0002-3123-3451}{}

\author{Shay Kutten}{Department of Industrial Engineering  and Management, Technion, Haifa, Israel}{kutten@technion.ac.il}{}{}

\author{Mordechai Shalom}{TelHai Academic College, Upper Galilee, 12210, Israel}{cmshalom@telhai.ac.il}{https://orcid.org/0000-0002-2688-5703}{}

\author{Shmuel Zaks}{Department of Computer Science, Technion, Haifa,  Israel, and School of Engineering, Ruppin Academic Center, Israel.}{zaks@cs.technion.ac.il}{https://orcid.org/0000-0001-5637-4923}{}

\authorrunning{Y. Emek et al.}

\Copyright{Yuval Emek and Shai Kutten and Mordechai Shalom and Shmuel Zaks}



\funding{This research was supported in part by PetaCloud - a project funded by the Israel Innovation Authority.}

\nolinenumbers 

\EventEditors{John Q. Open and Joan R. Access}
\EventNoEds{2}
\EventLongTitle{42nd Conference on Very Important Topics (CVIT 2016)}
\EventShortTitle{CVIT 2016}
\EventAcronym{CVIT}
\EventYear{2016}
\EventDate{December 24--27, 2016}
\EventLocation{Little Whinging, United Kingdom}
\EventLogo{}
\SeriesVolume{42}
\ArticleNo{23}

\begin{document}

\maketitle

\begin{abstract}
A matching of a graph is a subset of edges no two of which share a common vertex, and a maximum matching is a matching of maximum cardinality. 
In a $b$-matching every vertex $v$ has an associated bound $b_v$, and a maximum $b$-matching is a maximum set of edges, such that every vertex $v$ appears in at most $b_v$ of them.  
We study an extension of this problem, termed {\em Hierarchical b-Matching}. In this extension, the vertices are arranged in a hierarchical manner. 
At the first level the vertices are partitioned into disjoint subsets, with a given bound for each subset. 
At the second level the set of these subsets is again  partitioned into disjoint subsets, with a given bound for each subset, and so on. 
In an {\em Hierarchical b-matching} we look for a maximum set of edges, that will obey all bounds (that is, no vertex $v$ participates in more than $b_v$ edges, then all the vertices in one subset do not participate in more that that subset's bound of edges, and so on hierarchically). 
We propose a polynomial-time algorithm for this new problem, that works for any number of levels of this hierarchical structure.
\end{abstract}
 
\section{Introduction}\label{sec:intro}
A matching of a graph is a subset of its edges such that no two edges share a common vertex. 
The maximum matching problem is the problem of finding a matching of maximum cardinality in a given graph.
The maximum weighted matching problem is an extension of the problem for edge-weighted graphs in which one aims to find a matching of maximum total weight.
Both problems are fundamental in graph theory and combinatorial optimization.
As such, they have been extensively studied in the literature.
While pioneering works focused on the non-weighted bipartite case, later work considered general graphs and weights.
The general case is solved in \cite{Edmonds65-PathsTreesFlowers} and a more efficient algorithm is proposed later in \cite{MV80}.

An important extension of these problems is the following maximum $b$-matching problem. 
We are given a (possibly weighted) graph and a positive integer $b_v$ for every vertex $v$ of the graph.
A $b$-matching of the graph is a multiset $M$ of its edges such that,
for every vertex $v$, 
the number of edges of $M$ incident to $v$ does not exceed $b_v$.
Clearly, a matching is a special case of $b$-matching in which $b_v=1$ for every vertex $v$.
The problems of finding a $b$-matching of maximum cardinality and of maximum weight are widely studied.
The weighted version of the problem is solved in \cite{Pul73}.
A faster algorithm is later proposed in \cite{Anstee87}.
This result is recently improved in \cite{Gabow2018-b-matching}.
Being a fundamental problem, the $b$-matching problem is considered in the literature in specific graph classes (e.g. \cite{Tamir95}, \cite{DuCoffePopa18-b-Matching}).
The $b$-matching model is used to solve numerous problems in different areas, e.g. \cite{Tennenholtz2002TractableCombinatorialAuctions}.
See \cite{LP09} for an excellent reference for matching theory, and \cite{lawler1976combinatorial} for flow techniques and other algorithmic techniques used to solve matching problems.

In this work, we consider an extension of the maximum cardinality $b$-matching problem and propose a polynomial-time algorithm for it.
In this extension, the vertices are arranged in a hierarchical manner, 
such that every set is either a single vertex or the union of other sets,
and with each set, there is an associated upper bound on the sum of the degrees of the vertices in it.

This problem can arise in hierarchical structures, for instance as in the following scenario. 
Pairs of researchers are willing to pay mutual visits to each other.
However, every researcher $r$ has a budget that allows her to exercise at most $b_r$ visits.
The goal is to find a maximum number of such pairs (that will visit each other) without exceeding the individual budget of any researcher.
This problem can be modeled as a $b$-matching problem.
Now we extend this scenario to the hierarchical case.
Some institutions assign budgets not only to individual researchers but also to research groups, departments, faculties and so on.
In this case a set of pairs is feasible if the number of visits to be done by every individual researcher, every research group, every department and every faculty is within its assigned budget.

In Section \ref{sec:prelim} we introduce basic definitions, notations and the problem's statement. In Section \ref{sec:pseudoPolynomial} we present a pseudo polynomial algorithm for the problem, 
which is improved in Section \ref{sec:fullyPolynomial} to a polynomial-time algorithm. We conclude with remarks and further research in Section \ref{sec:summary}.
 
\section{Preliminaries}\label{sec:prelim}
\runningtitle{Sets}
For two non-negative integers $n_1,n_2$, denote by $[n_1,n_2]$ 
the set of integers that are not smaller than $n_1$ and not larger than $n_{2}$.
$[n]$ is a shorthand for $[1,n]$,
and $\uplus$ denotes the union operator of multisets, 
i.e., $A \uplus B$ is a multiset in which the multiplicity of an element is the sum of its multiplicities in $A$ and $B$.
A set system $\cl$ over a set $U$ of elements is \emph{laminar} if for every two sets $L, L' \in \cl$ one of the following holds: 
$L \cap L' = \emptyset$, $L \subseteq L'$, or $L' \subseteq L$.
We consider laminar set systems of distinct sets and $\emptyset \notin \cl$.
In this case, at most one of the conditions may hold.
Since adding $U$ and all the singletons of $U$ to $\cl$ preserves laminarity, 
we assume without loss of generality that $U \in \cl$ and $\set{u} \in \cl$ for every $u \in U$.

The next Lemma summarizes well known fact about laminar set systems:
\begin{lemma}
\begin{enumerate}[i)]
Let $\cl$ be a laminar set system over a set $U$. Then

\item the elements of $\cl$ can be represented as a rooted tree $T_{\cl}$, 
in which the root corresponds to $U$ and the leaves to the singletons $\set{u}$ for every $u \in U$.

\item \label{itm:AtLeastTwoChildren} Every non-leaf node of $T_{\cl}$ has at least two children. 

\item \label{itm:NumberOfLaminarSets} The number of sets in $\cl$ is at most $2|U|-1$.

\end{enumerate}

\end{lemma}

Note that in this lemma: \textit{\ref{itm:AtLeastTwoChildren})} follows from the fact that we assumed that a laminar set system consists of distinct sets, 
and \textit{\ref{itm:NumberOfLaminarSets})} follows from \textit{\ref{itm:AtLeastTwoChildren})}, with equality (i.e. the number of elements is exactly $2|U|-1$) if and only if $T_{\cl}$ is a full binary tree. 
See Figure \ref{fig:Instance} for an example.

Identifying the sets of $\cl$ with the nodes of $T_{\cl}$ and we say that
a set $L \in \cl$ is the \emph{parent} of $L' \in \cl$,
or that $L'$ is a \emph{child} of $L$.
Denote by $\children(L)$  the set of all children of $L$.

\runningtitle{Graphs} 
We use standard terminology and notation for graphs, see for instance \cite{D12}.
Given a simple undirected graph $G$, denote by $V(G)$ the set of vertices of $G$ and by $E(G)$ the set of the edges of $G$.
Denoting an edge between two vertices $u$ and $v$ as $uv$,
we say that the edge $uv \in E(G)$ is {\em incident} to $u$ and $v$, 
$u$ and $v$ are the endpoints of $uv$, 
and $u$ is adjacent to $v$ (and vice versa).
For directed graphs, the arc $uv$ is said to be from $u$ to $v$.
Denote by $\delta_G(v)$ the set of all edges incident to the vertex $v$ of $G$, i.e., $\delta_G(v) = \set{uv \in E(G)}$.
The number of these edges is the \emph{degree} $d_G(v)$ of $v$ in $G$.
When there is no ambiguity, the subscript $G$ is omitted and written simply as $\delta(v)$ and $d(v)$.
A \emph{walk} of a graph (resp. directed graph) $G$ is a sequence $u_0, e_1, u_1, e_2, \ldots, e_\ell, u_\ell$ where 
every $u_i$ is a vertex of $G$, every $e_i$ is an edge (resp. arc) of $G$, and $e_i = u_{i-1} u_i$ for every $i \in [\ell]$. 
A \emph{trail} is a walk whose edges are distinct, i.e. $e_i \neq e_j$ whenever $i \neq j$, $i,j \in [\ell]$.
A walk (resp. trail) is \emph{closed} if $u_0=u_\ell$ and \emph{open} otherwise.

\runningtitle{Hierarchical $b$-Matching}
\newcommand{\inst}{(G,\cl,b,c)}
Let $G$ be a graph and $L \subseteq V(G)$ a set of vertices of $G$ and
$M \subseteq E(G)$ a multiset over the edges of $G$
where $x_{e,M} \geq 0$ denotes the multiplicity of $e$ in $M$.
The degree of $L$ in $M$ is the sum of the degrees (in $M$) of its vertices, i.e. $d_M(L) \defined \sum_{v \in L} d_M(v)$. 
($d_M(v)$ is the degree of $v$ in the graph induced by the edges of $M$.)
Clearly, for a singleton $\set{v}$, we have $d_M(\set{v}) = d_M(v)$.
Let $\cl$ be a laminar set system over $V(G)$,
$c$ a vector of positive integers indexed by the edges of $G$,
and $b$ a vector of positive integers indexed by the sets in $\cl$.
A multiset $M \subseteq E(G)$ is a Hierarchical $b$-matching (or {\hbm} for short) of $\inst$ if
the multiplicity $x_{e,M}$ of every edge $e$ of $G$ is at most $c_e$,
and $d_M(L) \leq b_L$ for every $L \in \cl$.
In this work, we consider the following problem
\probl
{\mhbm}
{A quadruple $\inst$ where\\
$G$ is a graph,\\
$\cl$ is a laminar set system over $V(G)$,\\
$b$ is a vector of positive integers indexed by (the sets of) $\cl$, and\\
$c$ is a vector of positive integers indexed by the edges of $G$.}
{An {\hbm} $M$ of $\inst$ of maximum cardinality.}

Without loss of generality the vertex set of $G$ is $[n]$, 
$\cl=\set{L_1, \ldots, L_m}$ with $m > n$, $L_i=\set{i}$ for every $i \in [n]$, $L_m=[n]$ is the root of $\cl$, $b_{L_k} = b_k$ for every $k \in [m]$, and $c_e = c_{i,j}$ for every edge $e=ij$ ($i,j \in [n]$) of $G$.

Assume also without loss of generality that 
\begin{itemize}
\item $\max \set{b_{L'} | L' \in \children(L)} \leq b_L \leq sum_{L' \in \children(L)} b_{L'}$
for every non-leaf $L \in \cl$, and 
\item $c_{uv} \leq \min \set{b_u, b_v}$ for every edge $uv$ of $G$. 
\end{itemize}
In fact, if this is not the case, the vectors $b$ and $c$ can be modified as follows without affecting the set of feasible solutions.
First process the sets $L \in \cl$ in a preorder manner and set $b_L = b_{\parent(L)}$ whenever $b_L > b_{\parent(L)}$.
Then process the sets $L \in \cl$ in a postorder manner and set $b_L = sum_{L' \in \children(L)}$ whenever $b_L > sum_{L' \in \children(L)}$.

Given an {\hbm} $M$ of $\inst$, 
define the \emph{slackness} $s_{e,M}$ of an edge $e$ as $c_e - x_{e,M}$, 
and the \emph{slackness} $s_{L,M}$ of a set $L \in \cl$ as $b_L - d_M(L)$. 
Whenever no ambiguity arises, the name of the matching in the indices is omitted, and $x_e, s_e, s_L, d(v)$ is written instead of $x_{e,M}, s_{e,M}, s_{L,M}$, and $d_M(v)$, respectively.
A set $L \in \cl$ or an edge $e$ of $G$) is \emph{tight} in $M$ if its slackness in $M$ is zero.
A vertex $v$ is \emph{tight} in $M$ if there is at least one tight set $L \in \cl$ that contains $v$.

A \emph{matching} of $G$ is an {\hbm} of $\inst$ where $\cl$ consists of $V(G)$ and its singletons, 
$c_e=1$ for every edge $e$ of $G$, $b_{\{v\}}=1$ for every vertex $v$ of $G$ and $b_{V(G)}=\abs{V(G)}$.
A matching $M$ \emph{matches} (or \emph{saturates}) a vertex $v$ if $d_M(v)=1$ 
and it \emph{exposes} $v$ otherwise (i.e., $d_M(v)=0$). 
$M$ \emph{saturates} (resp. \emph{exposes}) a set of vertices $W \subseteq V(G)$, if $M$ saturates (resp. exposes) all the vertices of $W$. 
Denote by $V(M)$ (resp. $\exposed(M)$) the set of vertices matched (resp. exposed) by $M$. 

\section{Hierachical b-Matching}\label{sec:hierarchical}
Given a matching $M$ of a graph $G$, an $M$-\emph{augmenting path} of $G$ is a path of odd length that 
starts with an edge that is not in $M$ and its edges alternate between edges and non-edges of $M$.
It is well known that a matching $M$ is of maximum cardinality in a graph $G$ if and only if $G$ does not contain an $M$-augmenting path (Berge's Lemma \cite{Berge57}).
Since finding an augmenting path can be done in linear time,
this implies a polynomial-time algorithm that starts with any matching (e.g., the empty one) and improves it using augmenting paths until no such path is found.

Our design of the  polynomial-time algorithm for {\mhbm} consists of three parts.
We start by proving an analogous Lemma for H-matchings. This implies a pseudo-polynomial algorithm to augment a given {\hbm}.
Applying this algorithm starting from an empty {\hbm}, until an augmentation is impossible, implies a pseudo-polynomial algorithm for {\mhbm}.
This is done in Section \ref{sec:pseudoPolynomial}.
We then improve the result to get a polynomial-time algorithm for a single augmentation step,
and extend the technique introduced in \cite{Anstee87} to improve the overall algorithm to run in polynomial time. This is done in Section \ref{sec:fullyPolynomial}.

\subsection{A pseudo-polynomial algorithm}\label{sec:pseudoPolynomial}
We now present a pseudo-polynomial algorithm for the {\mhbm} problem. 
Our solution reduces the problem to the problem of finding a maximum cardinality matching of a graph using a pseudo-polynomial reduction.

Consult Figure \ref{fig:Instance} and Figure \ref{fig:Repr} for the following definition.
\begin{definition}
The \emph{representing graph} of an instance $\inst$ of {\mhbm} is the graph 
$\repr \inst=(\cl_T \cup \cl_B \cup \ce, E_{IN} \cup E_{UP} \cup E_E)$ where \begin{itemize}
    \item $\cl_B = \cup_{k = 1}^m L_{k,B}$, $\cl_T = \cup_{k = 1}^{m-1} L_{k,T}$, and every set $L_{k,B}$ and $L_{k,T}$ consists of $b_k$ vertices,
    \item $\ce = \cup_{ij \in E(G)} E_{i,j}$, and every set $E_{i,j}$ consists of $c_{i,j}=c_{j,i}$ vertices,
    \item $E_{IN}$ contains $b_k$ edges for every $k \in [m-1]$ connecting the $b_k$ vertices of $L_{k,B}$ with the $b_k$ vertices of $L_{k,T}$ so that $L_{k,T} \cup L_{k,B}$ induces a perfect matching,
    \item $E_{UP}$ contains $b_k \cdot b_{k'}$ edges between the $b_k$ vertices of $L_{k,T}$ and the $b_{k'}$ vertices of $L_{k',B}$ so that $L_{k,T} \cup L_{k',B}$ induces a complete bipartite graph, whenever $L_{k'}$ is the parent of $L_k$ in $\cl$.
    Moreover, $E_{UP}$ contains $c_{i,j} \cdot b_i$ edges between the $c_{i,j}$ vertices of $E_{i,j}$ and the $b_i$ vertices of $L_{i,B}$ for every edge $ij$ of $G$.
    \item {$E_E$} contains $c_{i,j}$ edges between the $c_{i,j}$ vertices of $E_{i,j}$ and the $c_{j,i} (=c_{i,j}$) vertices of $E_{j,i}$ 
    so that $E_{i,j} \cup E_{j,i}$ induces a perfect matching.
\end{itemize}
Let $M$ be an {\hbm} of $\inst$. The \emph{representing matching} $\repr(M)$ of $M$ is the matching $M'$ of $\repr \inst$ constructed as follows.
\begin{itemize}
    \item Start with the empty matching $M'$.
    \item Processing the edges $e$ of $G$ in some predefined order, for every
    edge $e=ij$ of $G$, add to $M'$ 
    \begin{itemize}
        \item the last $s_e$ edges of $E_E \cap (E_{i,j} \times E_{j,i})$, and
        \item the $x_e$ edges connecting the first $x_e$ vertices of $E_{i,j}$ to the first $x_e$ vertices of $L_{i,B}$ that are yet unmatched in $M'$.
    \end{itemize}
    Note that at this point (a) all the vertices of $\ce$ are matched by $M'$, and (b) the number of vertices of $L_{i,B}$ matched by $M'$ is $d_M(i)$ for every $i \in [n]$.
    \item Process the sets $L_i \in \cl \setminus \set{L_m}$ in some predefined postorder traversal of $\cl$. 
    \begin{itemize}
    \item Assume that the number of vertices of $L_{i,B}$ matched by $M'$ is $d_M(i)$ at the time it is processed.
    As was already mentioned, this assumption holds when a leaf (i.e., a singleton) of $\cl$ is processed.
    
    \item Add to $M'$ the perfect matching induced by the last $s_{L_i}$ vertices of $L_{i,B}$ and the last $s_{L_i}$ vertices of $L_{i,T}$.
    At this point, all the vertices of $L_{i,B}$ are matched by $M'$ and the number of unmatched vertices of $L_{i,T}$ is $d_M(i)$.
    
    \item Add to $M'$ the $d_M(i)$ edges that connect the first $d_M(i)$ vertices of $L_{i,T}$ with the first $d_M(i)$ vertices of $L_{j,B}$ 
    that are yet unmatched by $M'$ where $L_j$ is the parent of $L_i$.
    At this point all the vertices of $L_{i,B} \cup L_{i,T}$ are matched by $M'$. 
    
    \item We note that after all the children of a set $L_j$ have been processed, 
    the number of vertices of $L_{j,B}$ matched by $M'$ is $\sum_{j' | L_{j'} \in \children(L_j)} d_M(L_{j'})=d_M(L_j)$,
    i.e. our assumption holds when $L_j$ is processed.
    \end{itemize}
\end{itemize}
\end{definition}

\begin{figure}
\begin{center}
\scalebox{1}{\commentfig{
\includegraphics[width=\textwidth]{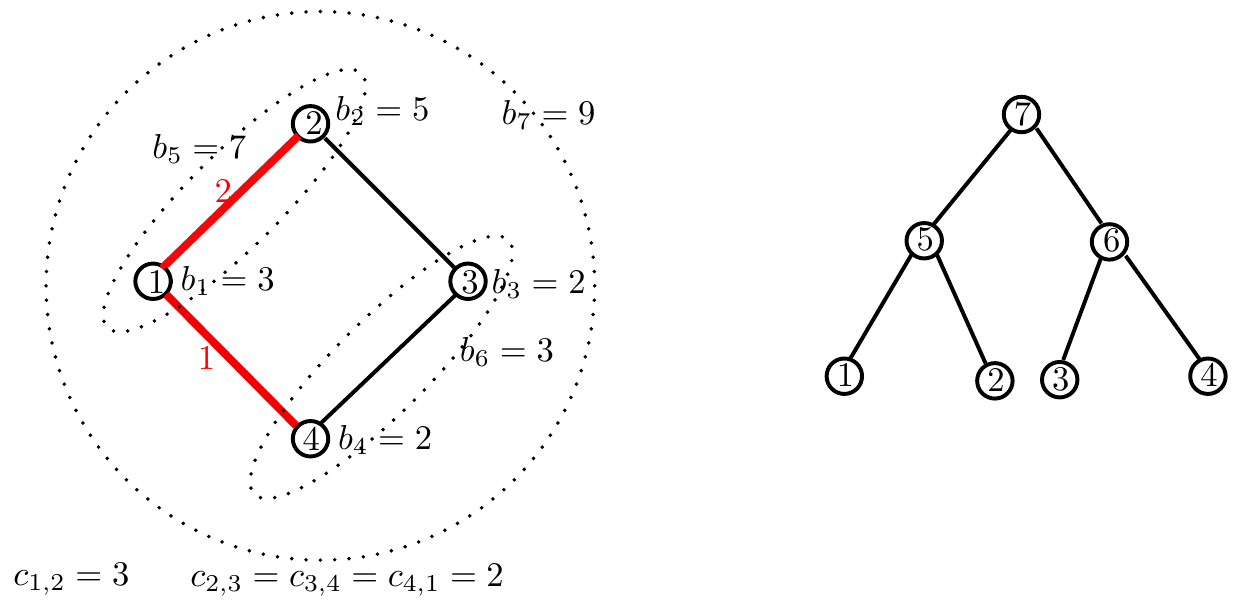}}}
\caption{An instance of {\mhbm} and a solution $M$: 
Depicted on the left hand side is the graph $G$ that is a cycle on four vertices 1, 2, 3, and 4. 
The set system $\cl$ consists of the four singletons $L_i=\set{i}$ for $i \in [4]$, and the sets $L_5=\set{1,2}$, $L_6=\set{3,4}$ and $L_7=V(G)$. 
The edges of $M$ are shown as thick red lines. 
The numbers on these lines are the multiplicities of the corresponding edges in $M$. 
Depicted on the right hand side is the Hasse diagram of the inclusion relation on $\cl$.}\label{fig:Instance}
\end{center}
\end{figure}

\begin{figure}
\begin{center}
\scalebox{1}{\commentfig{
\includegraphics[width=\textwidth]{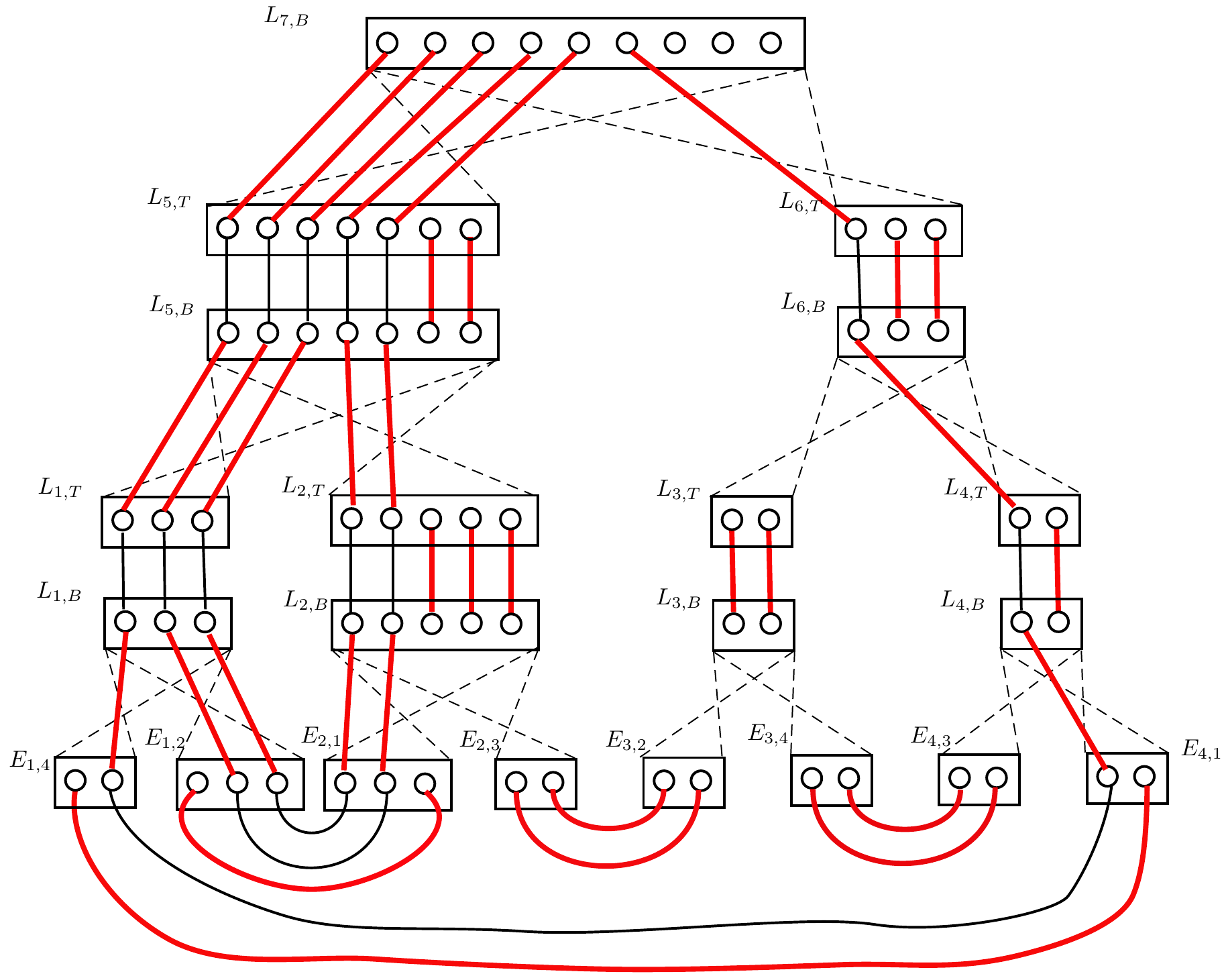}}}
\caption{The representing graph of the instance of {\mhbm} depicted in Figure \ref{fig:Instance}, and the representing matching $M'$ of its solution depicted in the same figure.
The edges of $M'$ are shown as thick red lines.
Two crossing dashed lines between two sets (e.g. the lines between $L_{1,T}$ and $L_{5,B}$) indicate that these sets are complete to each other.}\label{fig:Repr}
\end{center}
\end{figure}

The following Lemma is implied by the above description. 
\begin{lemma}\label{lem:BMatchingToMatching}
Let $M$ be an {\hbm} of $\inst$, and $M'=\repr(M)$. Then
\begin{enumerate}[i)]
    \item $\abs{V(M') \cap L_{m,B}} = d_M(L_m)=2
    \abs{M}$, and
    \item the number of edges of $M'$ between $L_{i,B}$ and $L_{i,T}$ is equal to $s_{L_i,M}$ for every $i \in [m-1]$.
\end{enumerate}
Moreover, $\exposed(M') \subseteq L_{m,B}$.
\end{lemma}
Informally, Lemma \ref{lem:BMatchingToMatching}  states that $(i)$ the number of matched vertices in the set $L_{m,B}$ is $2|M|$, and $(ii)$ the number of unmatched vertices of $M$ in $L_i$ is equal to the number of pairs matched by $M'$ between the sets $L_{i,B}$ and $L_{i,T}$ for every $i$, and in addition the unmatched vertices of $M'$ are all in the set $L_{m,B}$. 
In our example, the unmatched vertices are all in $L_{7,B}$, and $(i)$ the number of matched vertices in $L_{7,B}$ is $2|M|=6$. As for $(ii)$, 
for example, the slack of $L_1=\{1\}$ is $0$, and correspondingly there are no matched pairs between $L_{1,B}$ and $L_{1,T}$;  
the slack of  $L_2=\{2\}$ is $3$, and correspondingly there are $3$ pairs of $M'$ between $L_{1,B}$ and $L_{1,T}$; and the slack of  $L_6=\{3,4\}$ is $2$, and correspondingly there are $2$ pairs of $M'$ between $L_{6,B}$ and $L_{6,T}$. 

We now prove the opposite direction. 
Informally it states that starting from a given matching $M'$ of {\mhbm} 
in which the unmatched vertices all belong to  the set $L_{m,B}$, we can construct, in polynomial time, an {\hbm} $M$ of $\inst$, such that $M'=\repr(M)$, and $M'$ satisfies  properties $(i)$ and $(ii)$ of that Lemma.

\begin{lemma}\label{lem:MatchingToBMatching}
Let $\inst$ be an instance of {\mhbm}, and $M'$ be a matching of $\repr \inst$ such that $\exposed(M') \subseteq L_{m,B}$. 
Then there exists an {\hbm} $M$ of $\inst$ such that $M'=\repr(M)$, and it satisfies 
\begin{enumerate}[i)]
    \item \label{itm:numOfTopMatches} $\abs{V(M') \cap L_{m,B}} = 2 \abs{M}$,
    \item \label{itm:numOfInternalMatches} the number of edges of $M'$ between $L_{i,B}$ and $L_{i,T}$ is equal to $s_{L_i,M}$ for every $i \in [m-1]$.
    \end{enumerate}
Moreover, $M$ can be found in time linear in the size of $\repr \inst$.

\end{lemma}
\begin{proof}
Consider an edge $e=ij$ of $G$, and the sets $E_{i,j}, E_{j,i} \subseteq \ce$
of vertices of $G' = \repr \inst$. 
For every vertex $w$ of $E_{i,j}$ let $w'$ be its corresponding vertex in $E_{j,i}$.
Then exactly one of the following holds:
a) $w w' \in M'$,
b) $w$ is matched to a vertex of $L_{i,B}$ and $w'$ is is matched to a vertex of $L_{j,B}$ by $M'$.
Let $M$ be the multiset of edges of $G$ such that, for every edge $e=ij$ of $G$,
$x_{e,M}$ equals to the number of vertices $w$ of $E_{i,j}$ for which condition b) holds.
We claim that $M$ is an {\hbm} of $\inst$.

By definition of $M$, the multiplicity of an edge $e$ in $M$ is at most $c_e$.
For every $i \in [n]$, the number of vertices of $L_{i,B}$ matched to a vertex of $\ce$ is $d_M(i)$.
Let $z_B \in L_{i,B}$ and $z_T \in L_{i,T}$ for $i \in [m-1]$ such that $z_B z_T$ is an edge of $G'$. 
Then exactly one of the following holds:
a) $z_B z_T \in M'$,
b) $z_B$ is matched to some vertex of $L_{j,T}$ and $z_T$ is matched to some vertex of $L_{j',B}$ such that $L_j$ is a child of $L_i$ and $L_{j'}$ is the parent of $L_i$.
From the above facts, it follows by induction on the structure of $\cl$ that the number of vertices of $L_{i,B}$ matched to a vertex of $L_{j,T}$ for some $j$ is $d_M(L_i)$.
Therefore, for every $L_i \in \cl$ we have $d_M(L_i) \leq b_i$.
We conclude that $M$ is an {\hbm} of $\inst$.

It follows that the number of vertices of $L_{m,B}$ matched to a vertex of $L_{j,T}$ for some $j$ (i.e., the number of vertices of $L_{m.B}$ matched by $M'$) is $d_M(L_m)=2 \abs{M}$.
\end{proof}

\newcommand{\alg}{\textsc{FindRepresentingMatching}}

\begin{theorem}
$\alg$ is a pseudo-polynomial algorithm for {\mhbm}.
\end{theorem}
\begin{proof}
By Lemma \ref{lem:BMatchingToMatching} we have $\exposed(M') \subseteq L_{m.T}$. 
We observe that for any matching obtained by applying a sequence of augmenting paths to $M'$, in particular for $M'^*$ we have $\exposed(M'^*) \subseteq \exposed(M') \subseteq L_{m,B}$. 
Furthermore, $\exposed(M'^*)$ is minimum. 
Therefore, $V(M'^*) \cap L_{m,B}$ is maximum.
By Lemma \ref{lem:MatchingToBMatching} \textit{\ref{itm:numOfTopMatches})}, $\abs{M^*}$ is maximum.

The size of $G'$ is $\bigoh \left( \sum_{e \in G} c_e + \sum_{L \in \cl} b_L  \right) = \bigoh \left( \sum_{L \in \cl} b_L  \right)$. 
Since we want to prove pseudo-polynomial running time we assume that the values $b_L$ are bounded by $\abs{V(G)}^c$ for some fixed $c > 0$.
Then, $ \sum_{L \in \cl} b_L \leq \abs{\cl} \abs{V(G)}^c \leq 2 \abs{V(G)}^{c+1}$.
Therefore the size of $G'$ is $\bigoh(\abs{V(G)}^{c+1})$.
Since every step of the algorithm can be performed in time polynomial to $\abs{V(G')}$ we conclude the result.
\end{proof}

\alglanguage{pseudocode}
\begin{algorithm}[!ht]
\caption{$\alg$}\label{alg:PseudoPolynomial}
\begin{algorithmic}[1]
\Require{An instance $\inst$ of $\mhbm$}
\Ensure{Return an {\hbm} of $\inst$ of maximum cardinality.}
\Statex
\State $G' \gets \repr \inst$.
\State $M' \gets \repr(\emptyset)$.
\State $M'^* \gets$ a maximum matching of $G'$ such that $\exposed(M'^*)\subseteq \exposed(M')$.
\State $M^* \gets $ the {\hbm} of $\inst$ corresponding to $M'^*$ by Lemma \ref{lem:MatchingToBMatching}.  
\State \Return $M^*$.
\end{algorithmic}
\end{algorithm}

\subsection{Polynomial-time algorithm}\label{sec:fullyPolynomial}
In this section we improve the pseudo-polynomial algorithm of Section \ref{sec:pseudoPolynomial} to get a polynomial-time algorithm.
We achieve this in two stages.
First, we present (in Lemma 
\ref{lem:PolyTimeAugmentingPath}) a polynomial-time algorithm to augment a given matching $M$. 
Then, we present (in Lemma 
\ref{lem:anstee}) a technique to bound the number of augmentations by a polynomial in the size of the input. 
Combining these two lemmas we get (in Theorem \ref{theo:poly}) our polynomial-time algorithm. 

\begin{definition}
Let $M$ be an {\hbm} of $\inst$.
The \emph{augmentation graph} $\aug(M)$ of $M$ is the 2-edge-colored induced subgraph of $\repr \inst$ obtained by
\begin{itemize}
\item coloring every edge of $\repr(M)$ red, and every other edge blue,
\item marking $\min \set{2, d_M(L_i)}$ blue and $\min \set{2,s_{L_i}}$ red edges between $L_{i,T}$ and $L_{i,B}$ for every $i \in [m-1]$,
\item marking $\min \set{2, x_{e,M}}$ blue and $\min \set{2,s_{e,M}}$ red edges between $E_{i,j}$ and $E_{j,i}$ for every edge $e=ij$ of $G$, and finally
\item removing all vertices and edges that are not incident to any of the marked edges, except for two vertices $x_1,x_2$ of $L_{m,B}$ unmatched by $\repr(M)$.
\end{itemize}
\end{definition}
In our example, the first step of coloring the edges of $\repr(M)$ red and blue is depicted in 
Figure \ref{fig:Repr}, where the red edges are as indicated in that figure, and all other edges are blue. 

Though the above definition uses $\repr \inst$ in order to construct $\aug(M)$, 
it is easy to see that it can be constructed in time $\bigoh(\abs{E(G)})$, without constructing $\repr(M)$ at the first place.

\begin{lemma}\label{lem:PolyTimeAugmentingPath}
Let $M$ be an {\hbm} of $\inst$. Then
\begin{enumerate}[i)]
\item The only unmatched vertices in  $\aug(M)$ are $x_1$ and $x_2$, and 
\item $M$ is not of maximum cardinality if and only if $\aug(M)$ contains an alternating (odd) path (connecting $x_1$ and $x_2$).
\end{enumerate}
\end{lemma}
\begin{proof}
$i)$ follows immediately from the construction of $\aug(M)$.
We now prove $ii)$.
To show sufficiency, suppose that $\aug(M)$ contains a path $P$ as claimed.
Since $\aug(M)$ is a subgraph of $\repr \inst$, $P$ is a path of $\repr \inst$.
Moreover, since red edges are edges of $\repr(M)$ and blue edges are non-edges of it, $P$ is a $\repr(M)$-augmenting path.
Therefore, $M$ is not of maximum cardinality.

To show necessity, suppose that $M$ is not of maximum cardinality. 
Then, $H=\repr \inst$ contains a $\repr(M)$-augmenting path $P$.
We color those edges of $H$ that are in $\repr(M)$ red, and the others blue.
The path $P$ connects unmatched vertices, hence it must connect two vertices of $L_{m,B}$ and its end edges are blue.
Let $H'$ be the multigraph obtained by contracting every set $L_{i,X} (i \in [m], X \in \set{B,T})$ and $E_{i,j} (ij \in E(G))$ of vertices to a single vertex, 
and allowing parallel edges.
Clearly, $P$ corresponds to an alternating trail $T'$ of $H'$ starting and ending with blue edges incident to $L_{m,B}$.
Note that whenever $b_k=1$ for some $k \in [m-1]$ or $c_e=1$ for some edge $e=ij$ of $G$ the corresponding vertices, 
namely  $L_{k,B}$, $L_{k,T}$, $E_{i,j}$ and $E_{j,i}$ have multiplicity of one in the multigraph.
Therefore, such vertices may appear at most once in $T'$.
Let $T''$ be the shortest trail of $H'$ having these properties, and $v$ a vertex of $H'$.
The trail $T''$ does not contain even cycles, since by eliminating an even cycle one can get a shorter alternating trail. 
Therefore, the number of edges between any two occurrences of $v$ in $T''$ must be odd.
If $v$ occurs (at least) three times in $T''$ then two of the occurrences must be at even distance from each other. 
We conclude that the number of occurrences of $v$ in $T''$ is at most two.
We can construct from $T''$ a path $P''$ of $\aug(M)$ by 
a) splitting up every two parallel edges between some $L_{i,T}$ and $L_{i,B}$ into two disjoint edges,
b) splitting up every two parallel edges (that necessarily have the same color) between some $E_{i,j}$ and $E_{j,}$ into two disjoint edges, and 
c) splitting into two non-adjacent vertices any vertex that is still traversed twice by $P''$.
See Figure \ref{fig:SplitTrail} for this operation.
Note that the only case that the resulting path is possibly not a path of $\aug(M)$ is the case that a vertex $v$ is visited twice by $T''$ but has multiplicity of one in the multigraph.
However, as already observed, such vertices appear at most once in $T''$.
\end{proof}

\begin{figure}
\begin{center}
\scalebox{1}{\commentfig{
\includegraphics[width=0.7\textwidth]{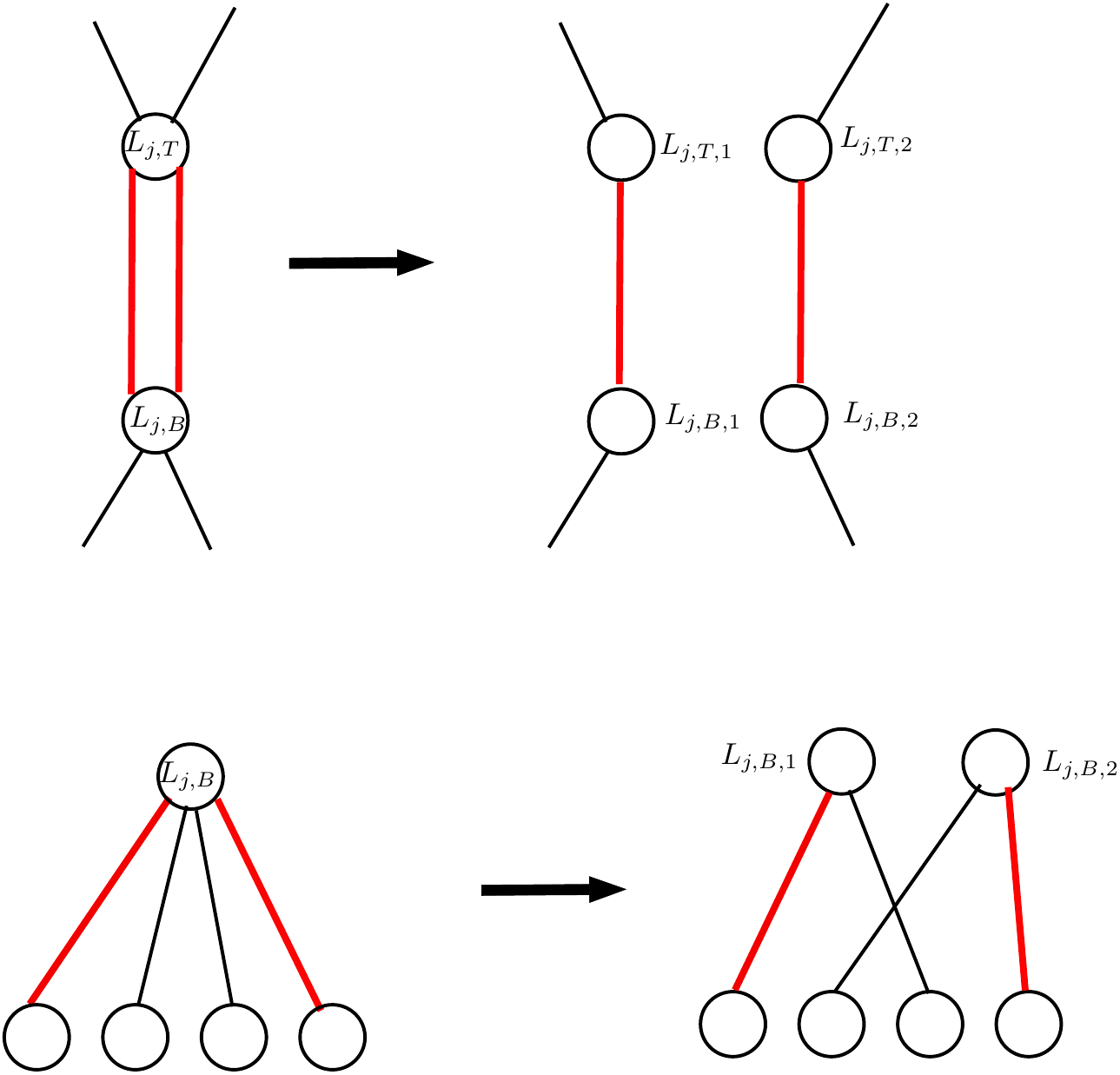}}}
\caption{The splitting of a trail $T''$ into a path $P''$.}\label{fig:SplitTrail}
\end{center}
\end{figure}

Since $\aug(M)$ can be constructed in time $\bigoh(\abs{E(G)})$ and an augmenting path of it can be found in time $\bigoh(\abs{E(G)})$,  
Lemma \ref{lem:PolyTimeAugmentingPath} implies that given an {\hbm} $M$, 
one can find in time $\bigoh(\abs{E(G)})$ an {\hbm} of cardinality $\abs{M}+1$ or decide that $M$ is of maximum cardinality. 
Therefore,

\begin{corollary}
A maximum cardinality {\hbm} $M$ of $\inst$ can be found in time $\abs{M} \cdot \bigoh(\abs{E(G)})$.
\end{corollary}

In the sequel we extend the technique introduced in \cite{Anstee87} to bound the number of augmentations needed to find a maximum matching by a polynomial.
This is done by finding a nearly optimal {\hbm} by means of a flow network.

\begin{lemma}\label{lem:anstee}
An {\hbm} $M$ of cardinality at least $\abs{M^*} - \bigoh(\abs{V(G)})$ can be found in time $\bigoh(\abs{V(G)} \cdot \abs{E(G)})$ where $M^*$ is a maximum cardinality {\hbm} of $\inst$.
\end{lemma}
\begin{proof}
We present an algorithm that works in three stages.
In the first stage we construct a flow network corresponding the instance,
and compute a maximum flow of it.
In the second stage we use this maximum flow to find an optimal fractional solution of the {\mhbm} instance.
Finally we round this fractional solution with a loss of $\bigoh(\abs{V(G)})$.

We start by describing the construction of the flow network $F=(N, A, \kappa, s, t)$ (see Figure \ref{fig:FlowNetwork}).
$N$ is a directed graph with vertex set $[m] \cup \set{i'| i \in [m]} \cup \set{s,t}$.
For every $k \in [m-1]$, there are two arcs $p k$ and $k' p'$
with capacity $\kappa(p k) = \kappa(k' p')= b_k$ where $L_p$ is the parent of $L_k$.
There are two arcs $s m$ and $m' t$, with capacities $\kappa(s m) = \kappa(m' t) = b_{L_m}$.
For every edge $e=ij$ of $G$, $N$ contains two arcs $i j'$ and $j i'$ with $\kappa (i j') = \kappa(j i') = c_e$.
Every {\hbm} $M$ of $\inst$ implies a feasible flow $f$ of $F$ by setting $f (i j')=f (j i')=x_{e,M}$ for every edge $e=ij$ of $G$, 
$f(p k)=f(k' p') = d_M(L_k)$ for every $k \in [m-1]$ such that $L_p = \parent(L_k)$,
and $f(s m) = f(m' t) = d_M(L_m)$.
It is easy to verify that $f$ is a feasible $s-t$ flow and its value $\abs{f}$ is $d_M(R)=2 \abs{M}$.
Therefore, for a maximum flow $f^*$ and an {\hbm} $M^*$ of maximum cardinality, we have $\abs{f^*} \geq 2 \abs{M^*}$.
The number of vertices of $N$ is $2 \abs{\cl} + 2$ which is linear in $n=\abs{V(G)}$.
The number of arcs of $N$ is dominated by $2 \abs{E(G)}$.
It is well known that, since all the capacities are integral, there is an  integral maximum flow $f^*$.
Such a flow can be found in time $\bigoh(\abs{V(G)} \cdot \abs{E(G)})$ \cite{Orlin2013}.
If $f^*$ is symmetric, i.e. $f^* (i j') = f^* (j i')$ for every edge $ij$ of $G$ then we may stop at this stage and $f^*$ induces an optimal solution of {\mhbm}. 

In the second stage we compute $\bar{f}_{a} = \frac{f^*(a)+f^*(a')}{2}$ for every pair of symmetric arcs $a$ and $a'$ of $N$. 
Clearly, $\bar{f}$ is symmetric and $\abs{\bar{f}}=\abs{f^*}$.
It remains to show that $\bar{f}$ is a feasible flow. Clearly the flow on each edge is bounded by its capacity; moreover, it is easy to verify that $\bar{f}$ satisfies flow conservation at each vertex $\neq s,t$, since $f^*$ satisfies it. 

If $\bar{f}$ is integral we can assign $x_{e,M}=\bar{f} (i j') = \bar{f} (j i')$ for every edge $e=ij$ of $G$ to get an optimal solution $M$ of {\mhbm}.
Since this is not necessarily the case, this assignment leads to a half-integral fractional {\hbm} $M'$ with $\abs{M'} \geq \abs{M^*}$.

In the last stage we round $M'$ with a small loss, to get a feasible solution $M$ with $\abs{M} \geq \abs{M'} - \bigoh(\abs{V(G)})$. 
We start with $M=M'$.
Let $H_M$ be the graph induced by the edges $e$ of $G$ such that $x_{e,M}$ is non-integral. 
As far as $M$ contains an even cycle $C$, we pick a matching of $C$ consisting of half of its edges, 
increase $x_{e,M}$ by $1/2$ for every edge $e$ of the matching and decrease it by $1/2$ for all the other edges of $C$.
This modification does not affect the degrees of the vertices and thus the degrees of the sets of $L$, 
in particular it does not affect $\abs{M}$. 
At this point all the cycles of $H_M$ are odd. 
These cycles are vertex disjoint, since otherwise there is an even cycle in $H_M$. 
Therefore, the number of these odd cycles is at most $\abs{V(G)}/3$.
Let $C$ be an odd cycle of $H_M$.
We pick a maximum matching of $C$, increase $x_{e,M}$ by $1/2$ for every edge $e$ of the matching and decrease it by $1/2$ for all the other edges of $C$.
This modification does not increase the degrees of the vertices and thus the degrees of the sets of $L$, 
however it decreases $\abs{M}$ by $1/2$.
Therefore, the total loss at this stage is at most $\abs{V(G)}/6$.
At this point $H_M$ is acyclic, i.e. a forest.
If we decrease $\abs{M}$ by $1/2$ in all the edges of the forest, $\abs{M}$ decreases by at most $\abs{V(G)}/2$.
We conclude that $\abs{M} \geq \abs{M'} - \frac{2}{3}\abs{V(G)} \geq \abs{M^*} -\frac{2}{3}\abs{V(G)}$.
\end{proof}

\begin{figure}
\begin{center}
\scalebox{1}{\commentfig{
\includegraphics[width=\textwidth]{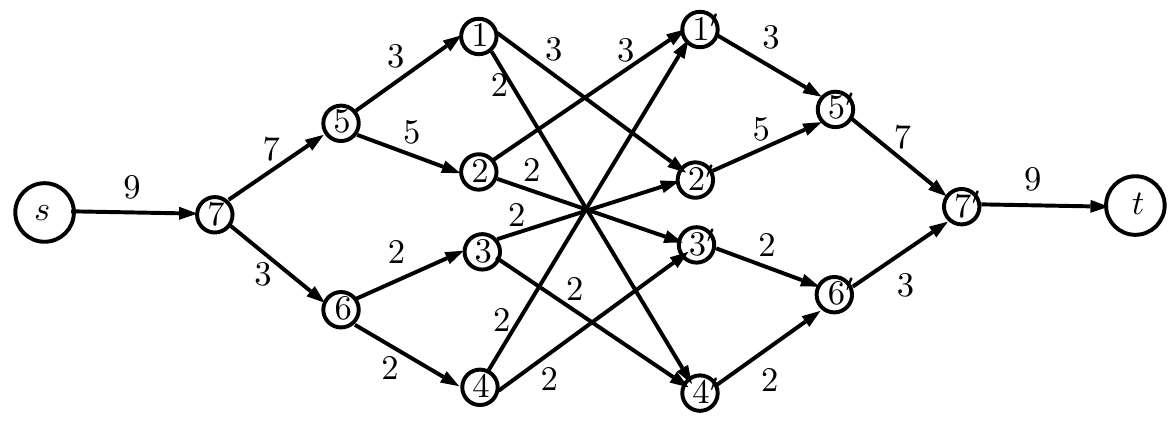}}}
\caption{The flow network instance corresponding to the instance of {\mhbm} depicted in Figure \ref{fig:Instance}.}\label{fig:FlowNetwork}
\end{center}
\end{figure}

We are now ready to prove our main theorem.
\begin{theorem}\label{theo:poly}
A maximum cardinality {\hbm} of $\inst$ can be found in time $\bigoh (\abs{E(G)} \cdot \abs{V(G)})$.
\end{theorem}
\begin{proof}
Let $M^*$ be a maximum cardinality {\hbm} of $\inst$.
By Lemma \ref{lem:anstee}, we can find an {\hbm} $M$ of cardinality at least $\abs{M^*} - \bigoh(\abs{V(G)})$ 
in time $\bigoh(\abs{V(G)} \cdot \abs{E(G)})$.
To get an optimal {\hbm} we can augment $M$ using $\bigoh(\abs{V(G)})$ augmentations, 
each of which can be done in time $\bigoh(\abs{E(G)})$, by Lemma \ref{lem:PolyTimeAugmentingPath}.
\end{proof}

\section{Conclusion and Open Problems}\label{sec:summary}
In this paper we studied the  {\hbm} problem, which is an extension of the $b$-matching problem, in which the vertices are organized in a hierarchical manner (independently of the structure of the graph). 
At each level the vertices are partitioned into disjoint subsets, with a given bound on the sum of degrees of every subset. 
The optimization problem is to find a maximum set of edges, that will obey all degree bounds.
This problem is applicable to many social structures, where the organization is of hierarchical nature. 

We have presented a polynomial-time algorithm for this new problem, in a few stages. 
We first reduced it to an ordinary matching in an associated larger graph, and this resulted in  a pseudo polynomial algorithm for the problem. 
We then improved it to a polynomial-time algorithm. 
This was achieved by combining results of two stages: in the first stage we presented a polynomial-time algorithm to augment a given matching, 
and in the second one we presented a technique to bound the number of augmentations by a polynomial in the size of the input.

A few open problems are immediately related to our result. The first one is the weighted case of the problem, in which every edge has an associated weight
and the goal is to find an {\hbm} of maximum weight. 
To find an efficient algorithm for this problem, apparently requires the extension of the linear programming techniques used in \cite{Pul73}.  
Another problem is to consider the hierarchical case where the bounds on the sets are interpreted as a bound on the number of edges that connect the vertices of the set to the rest of the graph.

\bibliography{GraphTheory,Mordo,Matching,References}

\newpage
\begin{thebibliography}{10}

\bibitem{Anstee87}
Richard~P. Anstee.
\newblock A polynomial algorithm for b-matchings: an alternative approach.
\newblock {\em Information Processing Letters}, 24:153--157, 1987.

\bibitem{Berge57}
C.~Berge.
\newblock Two theorems in graph theory.
\newblock {\em Proceedings of the National Academy of Sciences of the United
  States of America}, 43:842--844, 1957.

\bibitem{D12}
Reinhard Diestel.
\newblock {\em Graph Theory, 4th Edition}, volume 173 of {\em Graduate texts in
  mathematics}.
\newblock Springer, 2012.

\bibitem{DuCoffePopa18-b-Matching}
Guillaume Ducoffe and Alexandru Popa.
\newblock The b-matching problem in distance-hereditary graphs and beyond.
\newblock In {\em 29th International Symposium on Algorithms and Computation,
  {ISAAC} 2018, December 16-19, 2018, Jiaoxi, Yilan, Taiwan}, pages
  30:1--30:13, 2018.
\newblock URL: \url{https://doi.org/10.4230/LIPIcs.ISAAC.2018.30}, \href
  {http://dx.doi.org/10.4230/LIPIcs.ISAAC.2018.30}
  {\path{doi:10.4230/LIPIcs.ISAAC.2018.30}}.

\bibitem{Edmonds65-PathsTreesFlowers}
Jack Edmonds.
\newblock Paths, trees, and flowers.
\newblock {\em Canadian Journal of Mathematics}, 17:449–467, 1965.
\newblock \href {http://dx.doi.org/10.4153/CJM-1965-045-4}
  {\path{doi:10.4153/CJM-1965-045-4}}.

\bibitem{Gabow2018-b-matching}
Harold~N. Gabow.
\newblock Data structures for weighted matching and extensions to b-matching
  and f-factors.
\newblock {\em ACM Trans. Algorithms}, 14(3):39:1--39:80, June 2018.
\newblock URL: \url{http://doi.acm.org/10.1145/3183369}, \href
  {http://dx.doi.org/10.1145/3183369} {\path{doi:10.1145/3183369}}.

\bibitem{lawler1976combinatorial}
E.L. Lawler.
\newblock {\em Combinatorial Optimization: Networks and Matroids}.
\newblock Holt, Rinehart and Winston, 1976.
\newblock URL: \url{https://books.google.co.il/books?id=w\_lQAAAAMAAJ}.

\bibitem{LP09}
L.~Lov{\'a}sz and M.~D. Plummer.
\newblock {\em Matching Theory}.
\newblock AMS Chelsea Publishing, 2009.

\bibitem{MV80}
Silvio Micali and Vijay~V. Vazirani.
\newblock An $o(\sqrt{|V|}\cdot |e|)$ algorithm for finding maximum matching in
  general graphs.
\newblock In {\em Proceedings of the 21st Annual Symposium on Foundations of
  Computer Science}, SFCS '80, pages 17--27, Washington, DC, USA, 1980. IEEE
  Computer Society.
\newblock URL: \url{http://dx.doi.org/10.1109/SFCS.1980.12}, \href
  {http://dx.doi.org/10.1109/SFCS.1980.12} {\path{doi:10.1109/SFCS.1980.12}}.

\bibitem{Orlin2013}
James~B. Orlin.
\newblock Max flows in o(nm) time, or better.
\newblock In {\em Proceedings of the Forty-fifth Annual ACM Symposium on Theory
  of Computing}, STOC '13, pages 765--774, New York, NY, USA, 2013. ACM.
\newblock URL: \url{http://doi.acm.org/10.1145/2488608.2488705}, \href
  {http://dx.doi.org/10.1145/2488608.2488705}
  {\path{doi:10.1145/2488608.2488705}}.

\bibitem{Pul73}
R.~Pulleyblank.
\newblock {\em {Faces of Matching Polyhedra}}.
\newblock PhD thesis, University of Waterloo, 1973.

\bibitem{Tamir95}
Arie Tamir and Joseph~S.B. Mitchell.
\newblock A maximum b-matching problem arising from median location models with
  applications to the roommates problem.
\newblock {\em Mathematical Programming}, 80:171--194, 1995.

\bibitem{Tennenholtz2002TractableCombinatorialAuctions}
Moshe Tennenholtz.
\newblock Tractable combinatorial auctions and b-matching.
\newblock {\em Artificial Intelligence}, 140(1-2):231--243, September 2002.
\newblock URL: \url{http://dx.doi.org/10.1016/S0004-3702(02)00229-1}, \href
  {http://dx.doi.org/10.1016/S0004-3702(02)00229-1}
  {\path{doi:10.1016/S0004-3702(02)00229-1}}.

\end{thebibliography}

\end{document}